\documentclass[smallextended]{svjour3}       % onecolumn (second format)
\smartqed  % flush right qed marks, e.g. at end of proof
\usepackage{graphicx}

\usepackage{mathrsfs}
\usepackage{amsfonts}
\usepackage{amsmath}
\usepackage{fancyhdr}
\usepackage{textcomp}
\usepackage[official]{eurosym}
\usepackage{color}
\usepackage{enumitem}
\usepackage{amssymb}
\usepackage{caption}
\usepackage{multirow}
\usepackage{esint}
\usepackage{amssymb}
\usepackage{nicefrac}
\usepackage{mathrsfs}
\usepackage{amsfonts}
\usepackage{textcomp}
\usepackage{algorithm}
\usepackage{algorithmic}
\usepackage{latexsym}
\usepackage{dsfont}
\usepackage{color}
\usepackage{bbm}
\newcommand{\AVaR}{\text{AV@R}}
\newcommand{\VaR}{\text{V@R}}
\def \one {\mathchoice {\hbox{1\kern -0.27em l}}{\hbox{1\kern -0.27em l}}
                     {\small{1\kern -0.27em l}}{\small{1\kern -0.27em l}}}
\DeclareMathOperator*{\esssup}{ess\,sup}
\begin{document}

\title{The distortion principle for insurance pricing: properties, identification and robustness%\thanks{Grants or other notes
%about the article that should go on the front page should be
%placed here. General acknowledgments should be placed at the end of the article.}
}
%\subtitle{Do you have a subtitle?\\ If so, write it here}

%\titlerunning{Short form of title}        % if too long for running head

\author{Daniela Escobar         \and
        Georg Ch. Pflug %etc.
}

%\authorrunning{Short form of author list} % if too long for running head

\institute{Daniela Escobar \at
              University of Vienna. Department of Statistics and Operations Research (ISOR),
Oskar-Morgenstern-Platz 1, A- 1090 Wien-Vienna, Austria \\
              %Tel.: +123-45-678910\\
              %Fax: +123-45-678910\\
              \email{daniela.escobar@univie.ac.at}           %  \\
%             \emph{Present address:} of F. Author  %  if needed
           \and
           Georg Ch. Pflug \at
              ISOR and International Institute for Applied Systems Analysis (IIASA), Laxenburg,
Austria. \\
\email{georg.pflug@univie.ac.at}
}

\date{Received: February 27, 2018 / Accepted: date}
% The correct dates will be entered by the editor

\maketitle

\begin{abstract}
Distortion (Denneberg 1990)  is a well known premium calculation principle for insurance contracts. In this paper, we study sensitivity properties of distortion functionals w.r.t. the assumptions for risk aversion as well as robustness w.r.t. ambiguity of the loss distribution. Ambiguity is measured by the Wasserstein distance. We study variances of distances for probability models and identify some worst case distributions.  In addition to the direct problem we also investigate the inverse problem, that is how to identify the distortion density on the basis of observations of insurance premia.

% \PACS{PACS code1 \and PACS code2 \and more}
% \subclass{MSC code1 \and MSC code2 \and more}
\end{abstract}

\section{Introduction}\label{sec:1}

The function of the insurance business is to carry the risk of a loss of the customer for a fixed amount, called the premium. The premium has to be larger than the expected loss, otherwise the insurance company faces ruin with probability one.
The difference between the premium and the expectation is called the {\em risk premium}.
There are several principles, from which an insurance premium is calculated on the basis of the loss distribution.

Let $X$ be a (non-negative) random loss variable. Traditionally, an insurance premium is a functional, $\pi \: : \{  X\geq 0 \text{ defined on } (\Omega , \mathcal{F}, P) \} \rightarrow \mathbb{R}_{\geq 0}$. We will work with functionals that depend only on the distribution of the loss random variable (sometimes called law-invariance or version-independence property, Young 2014~\cite{young2014}). If $X$ has  distribution function $F$ we use the notation $\pi(F)$  for the pertaining insurance premium, and   $\mathbb{E}(F)$ for the expectation of  $F$. We use alternatively the notation $\pi(F)$ or $\pi(X)$, resp. $\mathbb{E}(F)$ or $\mathbb{E}(X)$ whenever it is more convenient. To the extent of the paper, a more specific notation is used for particular cases of the premium.\\

We consider the following basic pricing principles:
\begin{itemize}
\item The distortion principle (Denneberg 1990~\cite{denneberg1990}).
\item The certainty equivalence principle (v Neumann and Morgenstern 1947~\cite{neumann1947}).
\item The ambiguity principle (Gilboa and Schmeidler 1989~\cite{gilboa1989}).
\item Combinations of the previous (for instance Luan 2001~\cite{luan2001}).
\end{itemize}

{\bf The distortion principle.}
The distortion principle is related to the idea of stress testing. The original distribution function $F$ is modified (distorted) and the premium is the expectation of the modified distribution.  If $g:\, [0,1] \rightarrow \mathbb{R}$  is a concave monotonically increasing function with the property $g(0)=0$, $g(1)=1$, then the distorted distribution $F^{g}$ is given by
\[
F^{g}(x)=1-g(1-F(x)).
\]
The function  $g$ is called the {\em distortion function} and
\[
h(v)=g^{\prime}(1-v),
\]
with  $g^\prime$ being the derivative of $g$, is the {\em distortion density}.\footnote{The derivative of a concave function is a.e. defined, even if it is not differentiable everywhere.} Notice that $h$ is a density in $[0,1]$.  We denote by  $H(u)=\int_0^u h(v) \, dv$ the {\em distortion distribution}.
Since the assumptions imply that $g(x) \ge x$ for $0\le x \le 1$, $F^g \le F$, i.e. $F^g$ is first order stochastically larger than $F$.\footnote{ $F_1$ is first order stochastically larger than $F_2$ if $F_1(x)\le F_2(x)$ for all $x$.}  The distortion premium is the expectation of $F^{g}$
\[
\pi_h(F)=\int_{0}^{\infty}g(1-F(x))\,dx \ge \int_0^{\infty} (1- F(x)) \, dx = \mathbb{E}(X).
\]

By a simple integral transform, one may easily see that the premium can equivalently be written as
%\label{pih_inverse}
\begin{equation} 
\pi_h(F)=\int_{0}^{1}F^{-1}(v)\,h(v)\,dv = \int_0^1 \VaR_v(F) \, h(v)\, dv,
\end{equation}
%A functional of this form is called an L-functional. If the empirical distribution is considered instead of $F$, then this %functional is called an L-statistic (linear combination of order statistics).
where $\VaR_v(F) = F^{-1}(v)$, the quantile function. Note that a functional of this form is called an L-estimates (Huber 2011~\cite{huber2011}). If the random variable $X$ takes as well negative values, we could  generally define the premium as a  \textit{Choquet integral}
\begin{equation} 
\pi_h(F) = \int_{-\infty}^0 g(1-F(x)) - 1\, dx + \int_0^\infty g(1-F(x))\, dx.
\end{equation}

In principle, any distortion function which is monotonic and satisfies $g(u) \ge u$ is a valid basis for a distortion function. However, the concavity of $g$ guarantees that the pertaining distortion density $h$ is increasing, which - in insurance application - reflects the fact that putting aside risk capital gets more expensive for higher quantiles of the risk distribution. Nondecreasing distortion functions lead to non-negative distortion densities with the consequence that
$$\pi_h(F_1) \le \pi_h(F_2) \qquad \hbox{ whenver $F_2$ is stochastically larger than $F_1$. }$$
Relaxing the monotonicity assumption for $g$ would violate in general the monotonicity w.r.t. first stochastic order.
% and replacing it by the assumption of bounded variation would imply that $g$ is a difference of two monotonically increasing functions and hence, the premium would be a difference of two ordinary premia, violating in general the monotonicity w.r.t. first stochastic order.

{\bf Examples of distortion functions.}
Widely used distortion functions $g$ resp. the pertaining distortion densities $h$ are
\begin{itemize}
\item the power distortion with exponent $s$.
If $0<s< 1$,
\begin{equation}\label{ghforsles1}
g^{(s)}(v)=v^{s},\qquad h^{(s)}(v)=s(1-v)^{s-1}.
\end{equation}
The premium is  known as the Proportional Hazard transform (Wang 1995~\cite{wang1995}) and  calculated as
\begin{equation} 
\pi_{h^{(s)}}(F) = \int_0^\infty 1- F(x)^s \, dx = s\int_0^1 F^{-1}(v)(1-v)^{s-1} \, dv.
\end{equation}

If $s\ge 1$, then we take
\begin{equation}\label{ghforsbig1}
g^{(s)}(v)= 1- (1-v)^s, \qquad h^{(s)}(v) = s v^{s-1}.
\end{equation}
The premium is
\begin{equation} 
\pi_{h^{(s)}}(F) = \int_0^\infty 1- (1-F(x))^s \, dx = s\int_0^1 F^{-1}(v)v^{s-1} \, dv.
\end{equation}

If we consider integer exponent, the premium has a special representation.

\begin{proposition}
Let $X^{(i)}$, $i=1, \dots, n$ be independent copies of the random variable $X$, then the power distortion premium with integer power $s$ has the representation
$$\pi_{h^{(s)}} (X)= \mathbb{E}\left(\max\lbrace X^{(1)}, \dots, X^{(s)}\rbrace \right).$$
\end{proposition}

\begin{proof}
Let $F$ be the distribution of $X$. The power distortion premium for integer power $s$ is computed with $g^{(s)}$ in (\ref{ghforsbig1}) and by definition
$$\pi_{h^{(s)}}(F) =\int_0^\infty g^{(s)}(1-F(x))= \int_0^\infty 1 - F(x)^s \, dx.$$
The assertion follows from the fact that the distribution function of the random variable  $\max\lbrace X^{(1)}, \dots, X^{(s)}\rbrace$ is $F(x)^s$.
\end{proof}

Finally, notice that the distortion density is bounded for $s\ge 1$, but unbounded for $0<s<1$.
\item the Wang distortion or Wang transform (Wang 2000~\cite{wang2000})
\[
g(v)=\Phi\left(\Phi^{-1}(v)+\lambda \right),\qquad h(v)=\frac{\phi(\Phi^{-1}(1-v)+\lambda)}{\phi\left(\Phi^{-1}(1-v)\right)},\quad\lambda>0,
\]
where $\Phi$ is the standard normal distribution and $\phi$ its density.

\item the $\AVaR$ (average value-at-risk) distortion function and density are
\begin{equation}\label{ghcte}
g_\alpha(v)=\min\left\lbrace\frac{v}{1-\alpha} ,1\right\rbrace,\qquad h_\alpha(v)=\frac{1}{1-\alpha}\,\mathds{1}_{v\ge\alpha},
\end{equation}
where $0\leq \alpha<1$. The pertaining premium has different names, such as conditional tail expectation (CTE), CV@R (conditional value at risk) or ES (expected shortfall) (Embrechts et al. 1997~\cite{embrechts1997}). The premium is
\begin{equation} 
\pi_{h_\alpha}(F)=\int_{0}^{\infty}\min\left\lbrace\frac{1-F(x)}{1-\alpha},1\right\rbrace\,dx=\frac{1}{1-\alpha}\int_{\alpha}^{1}F^{-1}(v)\,dv.
\end{equation}
\item piecewise constant distortion densities. The insurance industry uses also piecewise constant increasing distortion functions.
For example, the following distortion function is used by a large reinsurer.
\begin{center}

\begin{tabular}{l|c|l|c}
$v$ & $h(v)$ & $v$ & $h(v)$ \\
\hline
[0,0.85] & 0.8443 & [0.988,0.992) & 3.6462 \\

[0.85,0.947) & 1.1731 & [0.992,0.993) &  4.0572 \\

[0.947,0.965) & 1.4121 & [0.993,0.996) & 6.5378 \\

[0.965,0.975) &  1.7335 & [0.996,0.996) & 12.7020 \\

[0.975,0.988) &  2.4806 & [0.996,1] & 14.9436. \\

\end{tabular}

\end{center}
\end{itemize}
For more examples on different choices of $h$ and also for different families of distributions, see Wang 1996~\cite{wang1996} and Furmann and Zitikis 2008~\cite{furman2008}.

{\bf Certainty Equivalence.}
Let $V$ be a convex, strictly monotonic disutility function.\footnote {The original notion of a utility function introduced by Neumann/Morgenstern was a concave monotonic $U$, such that the decision maker maximizes the expectation $\mathbb{E}(U(Y))$ of a profit variable $Y$. A disutility function can be defined out of a utility function by setting $V(u) = - U(-u)$.} The certainty equivalence premium is the solution of
\[
V(\pi)=\mathbb{E}(V(X)),
\]
i.e. it is obtained by equating the disutility of the premium and the expected disutility of the loss. The premium is written as follows
\[
\pi^V(F)=V^{-1}\left(\mathbb{E}(V(X))\right)= V^{-1} \left( \int_0^1 V\left(F^{-1}(v)\right) \, dv \right).
\]
By Jensen's inequality $\pi^V(F) \ge \mathbb{E}(F)$. Examples for disutilities $V$ are the power utility $V(x)=x^{s}$ for $s \ge 1$ or the exponential utility $V(x)=\exp(x)$.

Related to this premium, one could consider just the expected value and compute the expected disutility (Borch 1961~\cite{borch1961}) obtaining
\begin{equation}\label{expdisutility}
\pi(F) =  \mathbb{E}(V(X)).
\end{equation}
For generalizations of the CEQ premium see Vinel and Krokhmal 2017~\cite{vinel2017}.
{\bf The ambiguity principle. } Let $\mathfrak{F}$ be a family of distributions, which contains the "most probable" loss distribution $F$.
The ambiguity insurance premium is
\[
\pi^{\mathfrak{F}}(F) = \sup \{ \mathbb{E}(G) : G \in \mathfrak{F} \}.
\]
$\mathfrak{F}$ is called the {\em ambiguity set}. In an alternative, but equivalent notation, the ambiguity premium is given by
\begin{equation}\label{ambprem}
\pi^{\mathcal{Q}}(X)=\max\left\lbrace \mathbb{E}_{Q}(X)\,:\,Q\in\mathcal{Q}\right\rbrace ,
\end{equation}
where $\mathcal{Q}$ is a family of probability models containing the baseline model $P$. The functional inside the maximization needs not to be the expectation, but can be general, see e.g. Wozabal 2012~\cite{wozabal2012}, Wozabal 2014~\cite{wozabal2014}, Gilboa and Schmeidler 1989~\cite{gilboa1989} and our Section~\ref{sec:6}.

\begin{remark} In their seminal paper from 1989, Gilboa and Schmeidler~\cite{gilboa1989}  give an axiomatic approach to extended utility functionals of the form
\[
\min \left\{ \mathbb{E}_{Q}(U(Y))\,:\,Q\in\mathcal{Q}\right\},
\]
where $U$ is a utility function and $Y$ is a profit variable. For the insurance case, $U$ should be replaced by a disutility function $V$ and $Y$ should be replaced by a loss variable $X$ leading to an equivalent expression
\[
\max \left\{ \mathbb{E}_{Q}(V(X))\,:\,Q\in\mathcal{Q}\right\}.
\]
The link to (\ref{ambprem}) is obvious and it can be seen as a combination of expected disutility (\ref{expdisutility}) and ambiguity.
\end{remark}

\begin{remark}
Recall the fundamental pricing formula  of derivatives in financial markets states that
the price can be obtained by taking the maximum of the discounted expected payoffs, where the maximum is taken over all probability measures, which make the discounted
price of the underlying a martingale. This can be seen as an ambiguity price.
\end{remark}

The ambiguity premium is characterized by the choice of the ambiguity set $\mathfrak{F}$. In principle, this set can be arbitrary given as long as it contains $F$.
Convex premium functionals have a dual representation, which are also in the form of an ambiguity functional. For distortion functionals, this will be illustrated in the next section. Other important examples for  ambiguity premium prices can be defined through distances for probability distributions. Let $D$ be such a  distance, then an ambiguity set is given by
\[
\mathfrak{F}=\left\{ G:\,D(F,G)\le\epsilon\right\},
\]
with ambiguity premium
\[
\pi^\epsilon_{D}(F) = \max \{ \mathbb{E}(G) : D(F,G) \le \epsilon \}.
\]
We call $\epsilon$ the {\em ambiguity radius}. This radius quantifies not only the risk premium, but also the model uncertainty, since the real distribution is typically not exactly known and all we have is a baseline model $F$. In our Section~\ref{sec:6} we base ambiguity models on the Wasserstein distance $WD$.\\

{\bf Combined models.} Luan 2001~\cite{luan2001} introduced a combination of distortion and certainty equivalence premium prices by defining a variable $W$ distributed according to $F^g$  and setting
\[
\pi^V_h(F) = V^{-1}(\mathbb{E}[V(W)]) =  V^{-1} \left( \int_0^1 V\left(F^{-1}(v) \right) h(v) \, dv \right).
\]
Notice that $(F^g)^{-1} (v) = F^{-1}(1-g^{-1}(1-v))$. \\

More generally, one may also add ambiguity respect to the model and set
%Notice that $(F^g)^{-1}(1-g(1-v))=F^{-1}(v)$.
\begin{equation}\label{fullmodel}
\pi^{V,\, \epsilon}_{h}(F) = \sup \biggr{\lbrace} V^{-1} \left( \int_0^1 V\left(G^{-1}(v)\right) \,  h(v) \, dv \right) : D(F,G) \le \epsilon \biggr{ \rbrace}.
\end{equation}
Notice that (\ref{fullmodel}) contains all previous definitions by making some of the following parameter settings
$$h(v)= 1, V(v)=v, \epsilon=0.$$
If all parameters are set like that, we recover the expectation.

We could also consider the expected disutility premium (\ref{expdisutility})  and combine it with the distortion premium,  
\[
\int_0^1 V(F^{-1}(v))\, h(v) \,dv = \mathbb{E}[V(W)].
\]
%where $W$ is distributed according to $F^g$. 
Section~\ref{sec:6} will be dedicated to study the combination of  distortion and ambiguity premium prices.

As to notation, we denote by $\mathcal{L}^p$ the space of all random variables with finite $p$-norm for all $p\geq 1$
$$\| X \|_p = [\mathbb{E}(|X|^p)]^{1/p},$$ 
resp. $\| X \|_\infty = \esssup (|X|)$, the essential supremum. The same notation is used for any real valued function on $[0,1]$ and $p$ and $q$ are conjugates if $1/p + 1/q =1$.
% $h$ is a function on [0,1] with uniform probability.

\section{The distortion premium and generalizations}

The characterization and represestations of the distortion premium were  studied exhaustively. Among some of the most classic contributions we  mention the dual theory  of Yaari 1987~\cite{yaari1987}; and the characterization  by axioms of this premium developed in Wang, et al. 1997~\cite{wangetal1997}, where the power distortion for  $0<s<1$ is also characterized in a unique manner. A summary of other known  representations and new generalization of this premium will be presented below. Recall that any mapping $X \mapsto \pi(X)$ which is monotone, convex and fulfils translation equivariance\footnote{$\pi$ has translation equivariance property, if $\pi(X+c) = \pi(X) + c$, for $c\in \mathbb{R}$.} is a risk measure.  Furthermore, if $\pi$ is also positively homogeneous, monotonic w.r.t. the first stochastic order and subadditive\footnote{A premium $\pi$ is called subadditive, if $\pi(X+Y) \le \pi(X) + \pi(Y)$. Subadditivity and positive homogeneity imply convexity.}, then it is a \textit{coherent} risk measure (Artzner et al. 1999~\cite{artzner1999}). The distortion premium fulfils all these properties, therefore by the Fenchel-Moreau-Rockefellar Theorem, it has a dual representation.
%, which is given by the following proposition.

%The mapping $X \mapsto \pi_h(X)$ for a distortion premium is monotonic of first order, positively homogeneous and convex. Therefore, by the Fenchel-Moreau-Rockefellar Theorem, it has a dual representation, which is given by the following proposition.

\begin{theorem}\label{fenchm}(see Pflug 2006~\cite{pflug2006subdifferential}) The dual representation of the distortion premium with distortion density $h$ is given by
$$\pi_h(X) = \sup \{ \mathbb{E}(X \cdot Z) : Z= h(U), \hbox{ where $U$ is uniformly distributed on } [0,1] \}.$$
\end{theorem}

Note that all admissible $Z$'s in Theorem~\ref{fenchm} are densities on [0,1], since $h\ge 0$ and $\mathbb{E}(h(U))=1$. To put it differently, given $X$   defined on   $(\Omega, \mathcal{F}, P)$ and let $\mathcal{Q}$ be the set of all probability measures on $(\Omega, \mathcal{F})$ such that the density
$\frac{dQ}{dP}$ has distribution function $H$, the distortion distribution, then
$$\pi_h(X) = \sup \{ \mathbb{E}_Q (X) : Q \in \mathcal{Q} \}.$$
Therefore, every distortion premium can be seen as well as an ambiguity premium with $\mathcal{Q}$ as the ambiguity set.

Let us look into more detail to the special case of the $\AVaR$ premium. In this case, the dual representation specializes to
$$\pi_{h_\alpha}(X) = \sup \lbrace \mathbb{E}(X \cdot Z) : 0\le Z \le \frac{1}{1-\alpha};\,  \mathbb{E}(Z)=1 \rbrace.$$

From the previous representation, we can see that the $\AVaR$-distortion densities $h_{\alpha}$ are the extremes of the convex set of all distortion densities. This fact implies that any distortion premium can be represented as mixtures of $\AVaR$'s, such representations are called Kusuoka representations (Kusuoka 2001~\cite{kusuoka2001}, Jouini et al. 2006~ \cite{jouini2006}). Coherent risks have a Kusuoka representation of the form
$$\pi(F) = \sup_{K\in \mathcal{K}} \int_0^1  \AVaR_\alpha (F) \, dK(\alpha),$$

where $\mathcal{K}$ is a collection of probability measures in [0,1].
%One sees that the distortion premium is exactly a comonontone additive coherent risk.
In particular, for the distortion premium we have the following result (Pflug/R{\"o}misch 2007~\cite{pflug2007}).

 \begin{theorem} \label{Thmixavar} Any  distortion premium can be written as
 $$\pi_h(F) = \int_0^1 \AVaR_\alpha (X) \, dK(\alpha).$$
 The mixture distribution $K$ is given by the way how $h$ is represented as a mixture of the $\AVaR$-distortion densities, i.e.
 $$h(v) = \int_0^v \frac{1}{1-\alpha}  \, dK(\alpha).$$
\end{theorem}

%Using this proposition, one may easily see the following characterization.

%\begin{theorem} \label{Thcharacterization} If $\pi(\cdot)$ is version independent, positively homogeneous, monotonic w.r.t. the first stochastic order and comonotone additive,\footnote{$\pi$ is comonotone addive, if $\pi(F) = \pi(F_1)+\pi(F_2)$, whenever $F^{-1}=F_1^{-1}+F_2^{-1}$} then $\pi$ is a distortion premium.
%\end{theorem}
%Theorem~\ref{Thmixavar} proves that a mixture of $\AVaR_\alpha$'s with different $\alpha$-levels is again a distortion premium.
%$$\int_0^1 \AVaR_\alpha(F) \, dM(\alpha) = \int_0^1 F^{-1}(v) \, h(v) \; dv,$$
%with $h(v) = \int_0^v \frac{1}{1-\alpha} \, dM(\alpha)$.
The pure  $\AVaR_\beta$ is contained in this class by setting $K(\alpha) = \delta_\beta $, the Dirac measure at $\beta$. Moreover, the integral of the $\AVaR$'s is obtained for $K(\alpha) = \alpha $ and is defined as
\[
\int_0^1 \AVaR_\alpha(F) \, d\alpha= \int_0^1 F^{-1}(v) \, \left[ -\log(1-v) \right] \, dv,
\]
if the integral exists.

\begin{remark} Some other generalizations of the  distortion premium were studied in Greselin and Zitikis 2018~\cite{greselin2018}, where they consider a class of functionals
$$\int_0^1 \nu(\AVaR_\alpha(X), \AVaR_0(X))\, d\alpha,$$
with $\nu(\cdot,\cdot)$ an integrable function and show the Gini-index and Bonferroni-index belong to this class. These generalizations lead to \textit{inequality} measures instead of risk measures. 
\end{remark}
As a related generalization of the distortion premium one may consider
\begin{equation}\label{A}
R(X)=\int_0^1 \nu(\AVaR_\alpha(X)) \, k(\alpha) \, d\alpha,
\end{equation}
for some convex and  monotonic Lipschitz function $\nu$  and some non-negative function $k$ on [0,1]. Clearly, $R(X)$ is convex and monotonic, but in general is neither  positively homogeneous nor  translation equivariant unless $\nu$ is the identity (see Appendix for a proof). To our knowledge, functionals of the form (\ref{A}) are not used in the insurance sector. For this and some other  generalizations see the papers of Goovaerts et al. 2004~\cite{goovaerts2004} and Furman and Zitikis 2008~\cite{furman2008}.

%after in a remark in continuity
%Notice that the continuity properties of $F \mapsto A(F)$ can be easily derived from the continuity properties of $F \mapsto \AVaR_\alpha(F)$ as they are presented in detail in the next section. To our knowledge, functionals of the form (\ref{A}) are not used in the insurance sector.

\section{Continuity of the premium w.r.t. the Wasserstein distance}\label{sec:3}

In this section we study sensitivity properties of the distortion premium respect to the underlying  distribution. Some results in this section are related to those in Pichler 2013~\cite{pichler2013}, Pflug and Pichler 2014~\cite{pflug2014} and Kiesel et al. 2016~\cite{kiesel2016}. Similar results of continuity for variability measures  are studied in Furman et al. 2017~\cite{furmanetal2017}. To start, we recall the notion of the Wasserstein distance.

\begin{definition} Let  $(\Omega,d)$ be a metric space and  $P$, $\tilde{P}$ be two Borel probability measures on it. Then the Wasserstein distance of order $r\geq 1$ is defined as
\[
WD_{r,d} (P,\tilde{P})= \left( \inf_{X\sim P\atop Y\sim \tilde{P}} \mathbb{E}\left(d(X,Y)^r\right) \right) ^{1/r}.
\]
\end{definition}
Here the infimum is over all joint distributions of the pair $(X,Y)$, such that the marginal distributions are $P$ resp. $\tilde{P}$,  i.e. $X\sim P$, $Y \sim \tilde{P}$.

For two distributions $F$ and $G$ on the real line endowed with metric
$$d_1(x,y)= |x-y|.$$
this definition specializes to (see Vallender 1974~\cite{vallender1974})
\begin{eqnarray*}
WD_{1,d_1}(F,G) &=& \int_{-\infty}^\infty |F(x)-G(x)|\:dx
= \int_0^1 |F^{-1}(v) - G^{-1}(v)| \,dv.
\end{eqnarray*}

Therefore, the Wasserstein distance is the (absolute) area between the distribution functions which is also the (absolute) area between the inverse distributions. By a similar argument one may prove that the Wasserstein distance of order $r\geq 1$ with the $d_1$ metric on the real line is
\begin{equation} 
 WD^r_{r,d_1}(F,G)  = \int_{0}^1 |F^{-1}(v)-G^{-1}(v)|^r\:dv = \|F^{-1}-G^{-1}\|_r^r.
\end{equation}

We now study continuity properties of the functional $F \mapsto \pi_h(F)$.
\begin{proposition}[Continuity for bounded distortion densities] \label{contbound}  Let $F$ and $G$ be two distributions on the real line and $h$ a distortion density function. If the distributions have both finite first moments and $h$ is bounded, then
\[
\left| \pi_h(F) - \pi_h(G) \right| \leq || h ||_\infty \cdot WD_{1,d_1} (F,G).
\]
\end{proposition}

\begin{proof} See Pichler 2010~\cite{pichler2010}.
\end{proof}

\begin{remark}
The boundedness of $h$ is ensured if $g$ has a finite right hand side derivative at 0, and also if $g$ has finite  Lipschitz constant $L$, since $\|h\|_\infty \le L$.
\end{remark}

Proposition~\ref{contbound} can be easily generalized as follows.

\begin{proposition}[Continuity for distortion densities in $ \mathcal{L}^q$ for $q<\infty$]\label{contqbigg1} Let $F$ and $G$ be two distributions on the real line and $h$ a distortion density function. If $F$, $G $ have finite $p-$moments and $h\in \mathcal{L}^q$, then
\[ \left| \pi_h(F) - \pi_h(G) \right| \leq || h ||_q \cdot WD_{p,d_1} (F,G),\]
where $p$ and $q$ are conjugates.
% with $p\geq 1$:

\end{proposition}

\begin{proof} By  H{\"o}lder's inequality for $p$ and $q$ we obtain
\begin{align*}
| \pi_h(F) - \pi_h(G)| &= \left| \int_0^1  h(v)   \cdot \left( F^{-1}(v)-  G^{-1}(v)\right)  \: dv\right| \\
&\leq  \left(  \int _0^1 \left| h(v) \right|^q \,dv \right)^{1/q} \cdot \left(  \int _0^1 \left|  F^{-1}(v)-  G^{-1}(v)   \right|^p \, dv \right)^{1/p}    \\
&\leq|| h ||_q \cdot WD_{p,d_1} (F,G).
\end{align*}
\end{proof}

\begin{example} Let $F$ and $G$ be two distributions  with finite first moments.
\begin{itemize}
\item For the $\AVaR$ distortion premium $ ||h_\alpha||_\infty = \frac{1}{1-\alpha}$, and therefore
\[| \pi_{h_\alpha}(F) - \pi_{h_\alpha}(G)|\leq \frac{1}{1-\alpha} \cdot WD_{1,d_1} (F,G) . \]
\item For the power distortion with $s\ge 1$, $ ||h^{(s)}||_\infty= s$, and therefore
\[| \pi_{h^{(s)}}(F) - \pi_{h^{(s)}}(G)|\leq  s\cdot WD_{1,d_1} (F,G) . \]
\end{itemize}
\end{example}

The power distortion with $0<s<1$ is not bounded.  The next result is dedicated for this particular case.
\begin{proposition}[Continuity for the the power distortion with $0<s<1$]\label{contunboundpower}   Let  $F$ and $G$ be distribution functions  and $h^{(s)}$  the distortion density defined in (\ref{ghforsles1}).
If $F$ and $G$ have finite $p-$moments for  $p>\frac{1}{s}$ and $h\in \mathcal{L}^q$,
then
\[| \pi_{h^{(s)}}(F) - \pi_{h^{(s)}}(G)| \leq  \frac{ s}{ \sqrt[q]{1+q\, (s-1)} }  \cdot  WD_{p,d_1} (F,G),\]
where $p$ and  $q$ are conjugates.
\end{proposition}

\begin{proof} We first note that $p>\frac{1}{s}$ implies $q< \frac{1}{1-s} $ and let $t=1+q\, (s-1)>0$.

\begin{align*}
\left( \int_0^1  h^{(s)}(v)^q  \: dv  \right)^{1/q}&= \left(  \int_0^1 s^q \cdot(1-v)^{q\cdot (s-1)} \: dv  \right)^{1/q} \\
& =  \left( \int_0^1 s^q \cdot(1-v)^{t-1}  \: dv  \right)^{1/q}  \\
&= \frac{s}{ \sqrt[q]{t} } \cdot  \left( \int_0^1 t\,(1-v)^{t-1} \, dv\right)^{1/q}  = \frac{s}{ \sqrt[q]{t} }.
\end{align*}
Proposition~\ref{contqbigg1} proves the statement.
\end{proof}

The next result is a direct consequence of Proposition~\ref{contunboundpower}.

\begin{corollary}[Continuity for  distortion densities dominated by power distortion densities with $0<s<1$]\label{contdomhs} Let  $F$ and $G$ be distribution functions and $h$ a distortion density. If $h$ is such that $h(v)\leq c\cdot h^{(s)}(v)$, for all $v\in [0,1]$, $c>0$ and $0<s<1$,  $F$ and  $G$ have finite $p-$moments for $p>\frac{1}{s}$ , then   $h\in \mathcal{L}^q$ and
\[| \pi_h(F) - \pi_h(G)|  \leq  \frac{c\cdot s}{ \sqrt[q]{1+q\, (s-1)} } \cdot WD_{p,d_1}(F,G) ,\]

where   $p$ and $q$ are conjugates.
\end{corollary}

\begin{corollary}[Convergence]\label{convdist}  If  $F, F_n$ for all $n\geq 1$ have finite uniformly bounded $p-$moments, $h\in \mathcal{L}^q$ and $WD_{p,d_1} (F_n,F) \rightarrow 0$ as $n\rightarrow \infty$, then
\[ \left| \pi_h(F) - \pi_h(F_n) \right| \xrightarrow[n \to \infty]{} 0,\]
where $p$ and $q$ are conjugates.
\end{corollary}

\begin{remark}  Corollary~\ref{convdist} holds  when the sequence of distributions are the empirical distributions $\widehat{F}_n$ defined on an i.i.d. sample of size $n$, $(x_1, \dots , x_n)$ from $X\sim F$. If $F$ has finite $p-$moments, then  $WD_{p,d_1} (\widehat{F}_n,F)  \xrightarrow[n \to \infty]{} 0$, hence $ \left| \pi_h(\widehat{F}_n) - \pi_h(F) \right| \xrightarrow[n \to \infty]{} 0$. This result follows by applying Lemma 4.1 in \cite{pflug2014}.
\end{remark}

Finally notice that, for continuity, the order of the Wasserstein distance $r$ coincides with the number of finite moments of $F$. 

\subsection{Partial coverage}

Many insurance contracts do not guarantee complete indemnity, but their payoff is just a part of the full damage. Such contracts include proportional insurance, deductibles and capped insurance. In general, there is a (monotonic) payoff function $T$ such that the payoff is $T(X)$, if the total loss is $X$. A quite flexible form is for instance the excess-of-loss insurance (XL-insurance), which has a payoff function
\begin{equation} 
T(x)=\left\{ \begin{array}{ll}
               0 & \qquad \hbox{ if } x \le a\\
               x-a & \qquad \hbox{ if } a \le x \le e \\
               e-a & \qquad \hbox{ if } x \ge e.
               \end{array}
               \right.
\end{equation}

Denote by $F^T$ the distribution of $T(X)$, if $F$ is the distribution of $X$. The distortion premium for partial coverage is $\pi_h(F^T)$. We study the relationship between $F^T$ and $G^T$ as well as between $\pi_h(F^T)$ and $\pi_h(G^T)$ in a slightly more general setup, namely for H\"older continuous $T$. Recall that $T$ is H\"older continuous with constant $H_\beta$, if $|T(x)-T(y)|\le H_\beta \cdot |x-y|^\beta$, for some $\beta\le 1$.

\begin{theorem}[Distance between the original and image probabilities by $T$]\label{wdimageT}  Let $P$ and $Q$ be two probability measures  and consider their image probabilities  under $T$ denoted by $P^T$ and $Q^T$, respectively. If $T$ is a $\beta-$H{\"o}lder continuous mapping, then
\[ WD_{r_\beta, d_1} (P^T, Q^T) \leq H_\beta \cdot WD_{r,d_1}^\beta(P,Q),  \]
for $r_\beta =\frac{r}{\beta}\geq 1$ and $r\geq1$, where  $H_\beta $ is the $\beta-$H{\"o}lder constant.
\end{theorem}

\begin{proof} Let the joint distribution of $X$ and $Y$ such that
$$WD_{r,d_1}(X,Y) = \mathbb{E}^{1/r}(|X-Y|^r),$$
then
\begin{align*}
 WD^{ r_\beta }_{r_\beta, d_1} (P^T, Q^T) &\le \mathbb{E}(|T(X) - T(Y)|^{r_\beta}) \\
 &\le H_\beta^{r_\beta} \cdot \mathbb{E}(|X - Y|^r)  =  H_\beta^{r_\beta} \cdot   WD^{r}_{r,d_1} (P,Q).
 \end{align*}
Taking the $r_\beta$ root on both sides finished the proof.
\end{proof}
For the XL-insurance, the H\"older-constant is a Lipschitz constant ($\beta=1$) and has the value 1.

From the previous Theorem we can conclude that, if two probabilities are close, then the image probabilities by a mapping $T$ with the characteristics of Theorem~\ref{wdimageT}, are close in Wasserstein distance as well. Theorem~\ref{wdimageT} isolates the argument also used in Theorem 3.31 in \cite{pflug2014}. Note that the underlying distances for the Wasserstein distances are the metrics of the respective spaces.

\begin{corollary} Let $F,G$ be two distributions defined by the probabilities $P$ and $Q$, respectively, and $F^T, G^T$ be their image distributions by $T$, respectively. If $T$ is a $\beta-$H{\"o}lder continuous mapping with constant $H_\beta$,  $h\in \mathcal{L}^q$, the distributions $F^T$, $G^T$ with finite $p-$moments, then for all $r=p\cdot \beta$ ($r\geq 1$), the distortion premium   with payment function $T$ satisfies
\begin{equation} 
| \pi_{h}(F^T) - \pi_{h}(G^T) |\leq || h||_q \cdot WD_{p,d_1}(P^T,Q^T)\leq || h||_q \cdot H_\beta\cdot WD_{r,d_1}^{\beta}(P,Q).
\end{equation}
\end{corollary}

We proceed now to study sensitivity properties of the distortion premium w.r.t. the distortion density.

\section{Continuity of the premium w.r.t. the distortion density}

Previously, we studied the mapping $F \mapsto \pi_h(F)$ for fixed $h$. In this section, we consider and present properties of the mapping $h \mapsto \pi_h(F)$ for fixed $F$. Different sensitivity properties w.r.t. the distortion parameters were studied in Gourieroux and Liu 2006~\cite{gourieroux2006}.  

\begin{proposition}[Continuity of the distortion premium  w.r.t. the distortion density $h$]\label{contdistfixh} Let $F$ be a distribution and consider two different distortion densities $h_1, \, h_2$. If $F$ has finite $p-$moments and $h_1, h_2\in \mathcal{L}^q$, then
\[
\left| \pi_{h_1}(F) - \pi_{h_2}(F) \right|\leq  ||F^{-1}||_p \cdot  || h_1-h_2 ||_q ,
\]
where $p$ and $q$ are conjugates. Here the choices $p=1$, $q=\infty$ and $p=\infty$, $q=1$ are included.
\end{proposition}
\begin{proof} Use  H\"older inequality and the result is direct.
\end{proof}

We can conclude that, if $h_1$ and $h_2$ are close, then also the premium prices are close. However, $h$ is always identifiable by the following Proposition.

\begin{proposition} If $\pi_{h_1}(F) = \pi_{h_2}(F)$ for all distribution functions $F$ (the value $\infty$ is not excluded), then
$$h_1(v) = h_2(v) \mbox{   a.s.}$$
\end{proposition}

\begin{proof}
Let $F_a$ be the distribution which takes the value 0 with probability $a$ and the value 1 with probability $1-a$, for some $a\in (0,1)$, then its inverse $F_a^{-1}$ is the
indicator function of the interval $[a,1]$. Hence,
\[
\pi_{h_1}(F_a) =  \int \mathds{1}_{[a,1]}(v) \, h_1(v) \, dv = \int_a^1 h_1(v) \, dv =  \pi_{h_2}(F_a) =\int_a^1 h_2(v) \, dv.
\]
%\begin{eqnarray*}
%\pi_{h_1}(F_a)&=& \int \mathds{1}_{[a,1]}(v) \, h_1(v) \, dv = \int_a^1 h_1(v) \, dv \\
%&=&  \pi_{h_2}(F_a) =\int_a^1 h_2(v) \, dv.
%\end{eqnarray*}
Thus, the distortion distributions $H_1$ and $H_2$ are equal and therefore $h_1 = h_2$ almost surely.
\end{proof}

\begin{remark} Note the previous proposition is true if the family of distributions where the premium prices coincide contains all the Bernoulli variables. Compare also  Theorem 2 in~\cite{wangetal1997}.
\end{remark}

\begin{remark}
Another family with the property that the premium prices for this family determine the distortion in a unique manner is the family of Power distributions of the form $F_{\gamma}(u)=u^\gamma$ on $[0,1]$ and more general of the form $F_{\gamma, \beta}(u)=\beta^{-\gamma}u^\gamma$ on $[0, \beta]$. The distortion premium prices for this family are
$$\int_0^1 \beta\,  v^{1/\gamma} \, h(v) \, dv, $$
and the uniqueness of $h$ and $\beta$ is obtained since 
\[
\beta =  \lim_{\gamma \rightarrow \infty } \int_0^1 \beta v^{1/\gamma} \, h(v) \, dv ,
\]
and the inversion formula for the Mellin transform (see Zwillinger 2002~\cite{zwillinger2002}).
\end{remark}

\section{Estimating the distortion density from observations}

The way how insurance companies calculate a premium is typically not revealed to the customer. Notice that risk premia appear not only in the insurance business, see the link of insurance premium prices and asset pricing in Nguyen et al. 2012~\cite{nguyen2012}. Risk premia appears in other areas such as

\begin{itemize}
\item {\bf power future markets.} A future contract fixes the price today for delivery of energy later. There is the risk of price changes between now and the delivery period. Thus, such a contract has the character of an insurance and the pricing principles apply, although the price is found in exchange markets (e.g. electricity future markets).
\item {\bf exotic options.} While standard options are priced through a replication strategy argument, this argument does not apply for other types of options and these options have the character of insurance contracts. Pricing of such contracts is often done over the counter, but again the pricing principle is not revealed to the counterparty.
\item {\bf credit derivatives.} Also these contracts carry the character of insurance and can be priced according to insurance price principles.
\end{itemize}

In this section we assume that we know the distortion premium prices of $m$ contracts, which are all priced with the same distortion density $h$. For each contract $j$, we also have a sample $x_1^{(j)}, \dots, x_n^{(j)}$ of size $n$ drawn from the loss distribution of this contract at our disposal. For simplicity we assume that $n$ is the same for all contracts, but this is not crucial.

The goal of this section is to show how the distortion density $h$ can be regained from the observations of the insurance prices, which would help us to shed more light on the price formation of contract counterparties. Notice that our aim is not to estimate the distortion premium prices from empirical data as is done in Gourieroux and Liu 2006~\cite{gourieroux2006} or Tsukahara 2013~\cite{tsukahara2013}.  \\

\textbf{A simulation example.}
As an example we consider $m$ {\em different} loss distributions, all of Gamma type. From each distribution, we obtain a sample of size $n$. For each sample, we calculate the $\AVaR$ and power distortion premium prices. Based on the prices obtained and our samples, we aim to recover the distortion density $h$.
We denote the ordered sample from the $j$-th loss distribution by $x^{(j)}_{[1]} , \dots , x^{(j)}_{[n]}$.  The distortion  premium, with distortion density $h$ for each sample $j=1, \dots , m$, is
\begin{equation}\label{pihdirect}
\pi^{(j)} = \sum_{i=1}^n x^{(j)}_{[i]}\int_{\frac{i-1}{n}}^{\frac{i}{n}} h(v)\: dv =\sum_{i=1}^n x^{(j)}_{[i]} \left( H\left( \frac{i}{n}\right) - H\left(\frac{i-1}{n} \right) \right)  .
\end{equation}

On the following, we develop (\ref{pihdirect}) for the particular cases of   $\AVaR$ and power distortion premium prices for each sample $j=1, \dots , m$.

%\textbf{$\AVaR$ distortion premium.}
%\begin{itemize}
\textbf{$\AVaR$ distortion premium.} The price for  $h_{ \alpha}$ defined on (\ref{ghcte}) is
\begin{equation}\label{directpiaveru}
\pi^{(j)}= \frac{1}{n\, (1-\alpha)} \cdot  \sum_{i=i_a}^n x^{(j)}_{[i]},
\end{equation}

where $1<i_\alpha<n$ s.t. $ \frac{i_\alpha -1}{n}\leq \alpha< \frac{i_\alpha}{n}$. \\

\textbf{Power distortion premium.} The price given by the power distortion $h^{(s)}$ defined in (\ref{ghforsles1})   with $0<s<1$ is
\begin{equation}\label{directpisles1}
\pi^{(j)} =  \sum_{i=1}^n x^{(j)}_{[i]}\cdot  \left(\left(1-\frac{i-1}{n}\right)^s  - \left( 1- \frac{i}{n}\right)^s  \right),
\end{equation}

and the price given by  $h^{(s)}$  defined in (\ref{ghforsbig1})  with $s\geq 1$ is
\begin{equation}\label{directpisbigg1}
\pi^{(j)}  =  \sum_{i=1}^n  x^{(j)}_{[i]} \cdot  \left(\left( \frac{i}{n}\right)^s - \left( \frac{i-1}{n}\right)^s  \right).
\end{equation}
%\end{itemize}
The inverse problem consists on estimating the distortion density $h$ from  observed prices. Recall that among the examples we presented of common distortion densities we had step functions and continuous functions,  therefore we will use step and spline functions in order to estimate estimate  $h$. We do so  for  the  prices obtained in (\ref{directpiaveru}), (\ref{directpisles1}) and  (\ref{directpisbigg1}).

\subsection{Estimation of the distortion density with a step function}

\textbf{Distortion density as a step function.} Let $\widehat{h}^1_l$ denote the step function consisting of $l$ equal-size steps, defined as
\begin{equation} 
\widehat{h}^1_l(v)= \sum_{k=1}^l \lambda_k \cdot I_{\left[L\cdot  \frac{k-1}{n}, L\cdot  \frac{k}{n}\right) }(v)=\sum_{k=1}^l \lambda_k \cdot I_{[\frac{k-1}{l},\frac{k}{l}) }(v),
\end{equation}

where $L= n/l$,  $ \lambda_s\in \mathbb{R}$ for $k=1,\dots, l$ and  $l$ denotes the dimension of the step function space. We also impose 
\begin{equation} 
 \int_0^1 \widehat{h}^1_l (v) \: dv=\sum_{k=1}^l \int_{ \frac{k-1}{l}}^{ \frac{k}{l}} \lambda_k \:dv =  \frac{1}{l}\cdot \sum_{k=1}^l \lambda_k = 1,
\end{equation}
with  $0\leq \lambda _1\leq \cdots \leq \lambda_l. $ In this way,  $\widehat{h}^1_l$ fulfils the density constraints as well as the non-decreasing constraints.\\

\textbf{Prices with the step function.} For each sample $j=1, \dots , m$, the prices with  $\widehat{h}^1_l$ are
\begin{equation}\label{pihatstepj}
 \widehat{\pi}^{(j)} = \sum_{i=1}^n x_{[i]}^{(j)} \cdot \int_{\frac{i-1}{n}}^{\frac{i}{n}} \widehat{h}^1_l (v) \: dv =\sum_{k=1}^l \sum_{i=k}^{L\cdot k} x_{[i]}^{(j)} \cdot  \int_{\frac{i-1}{n}}^{\frac{i}{n}} \lambda_k \: dv = \sum_{k=1}^l   \frac{\lambda_k}{n} \cdot \sum_{i=k}^{L\cdot k} x_{[i]}^{(j)} ,
\end{equation}

\textbf{Estimation.}  In order to estimate $\widehat{h}^1_l$ we will minimize the squares of the differences between the prices obtained by a distortion function $h$ and the premium obtained by $\widehat{h}^1_l$ in  (\ref{pihatstepj}). We will test our results with the given prices   $\pi^{(j)}$  calculated  in (\ref{directpiaveru}), (\ref{directpisles1}) and (\ref{directpisbigg1}). We solve, 
\[
\begin{aligned}
& \underset{\lambda}{\min}
& & \sum_{j=1}^m (\hat{\pi}^{(j)}-\pi^{(j)})^2 \\
& \text{s.t.} & &  \frac{1}{l}\cdot \sum_{i=1}^l \lambda_i = 1 \\
& & &   0\leq \lambda _1\leq \cdots \leq \lambda_l.
\end{aligned} \tag{$P_1$}\label{invphd}
\]

\subsection{Estimation of the distortion density with a cubic monotone spline}

\textbf{B-splines construction.} For our purposes we will define the splines on the interval $[0,1]$. Any B-spline is a linear combinations of the B-spline basis functions. The B-spline basis functions have all the same degree, $b$ and we choose to define them  at equally spaced knots  $t_k=k/L$, for $k=0, \dots , L$, hence $L$ subintervals. The functions for this basis  are denoted as $B_{k,b}$ and constructed following a recursion formula. The B-spline basis function of degree $0$ is denoted and defined as
\[ B_{k,0}(v)=\begin{cases}
1 & t_k \leq v \leq t_{k+1} \\
0 & \text{otherwise.}
\end{cases} \]
The B-spline basis functions of degree $b$, $B_{k,b}$ are obtained as an interpolation between  $B_{k,b-1}$ and $B_{k+1,b-1}$,  following the recursion formula
\[ B_{k,b}(v) = \frac{v-t_k}{t_{k+b} -t_k} B_{k, b-1}(v)  + \frac{t_{k+b+1}- v}{t_{k+b+1} -t_{k+1}} B_{k+1, b-1}(v).\]

 In the recursion we need to define fake knots $t_{-k}=0$ and $t_{L+k}=1$ for $k=1, \dots ,b$. In our case, we  consider splines of degree $b=2$.  If we divide $[0,1]$ in $L$ equally sized intervals, the basis has  $L+2 $ functions
\begin{equation}\label{Bbase2L}
\{ B_{-2,2}, B_{-1,2}, B_{0,2}, B_{1,2},\dots ,B_{L-1,2} \}.
\end{equation}
Notice that all the elements of the basis can be obtained by translating the B-spline basis function $B_{0,2}$ defined on the first $b+2=4$ knots. In order to have a base of increasing  monotone cubic  splines we integrate the functions of (\ref{Bbase2L}) and obtain a new base
\begin{equation}\label{Sbase2L}
\{S_{-2},S_{-1},S_0,\dots , S_{L-1}\},
\end{equation}
where $S_k(v)= \int_0^v B_{k,2}(w)\, dw$ for all $k=-2,\dots , L-1$. We scale the functions of (\ref{Bbase2L}) so the  splines in (\ref{Sbase2L}) are  distribution functions. Note that no linear combination of (\ref{Sbase2L}) gives us a constant function, due to construction of  (\ref{Sbase2L}). Therefore, we need one element to our base, say $S_{L}(v)=c$ and hence
\begin{equation}\label{Scbase2L}
\{S_{-2},S_{-1},S_0,\dots , S_{L-1}, S_L\},
\end{equation}
is our final base with $l=L+3$ elements, where $l$  denotes its dimension.\\

As an example we illustrate the base obtained for $L=5$. Starting with $B_{0,2}$  defined on  $t_0=0, t_1=1/5, t_2=2/5, t_3=3/5$, precisely
$$
B_{0,2}(v) =  \frac{5^3}{2} \cdot \left(  v^2   \mathds{1}_{[t_0,t_1)} + \left( v(t_2-v)+(t_3-v)(v-t_1) \right)   \mathds{1}_{[t_1,t_2)} + (t_3-v)^2  \mathds{1}_{[t_2,t_3)}\right)
$$
 We denote by $S_0$ the distribution of $B_{0,2}$ and obtain  the rest of the monotone cubic splines   by translating $S_0$. The basis of cubic monotone splines of dimension $l=8$, illustrated in Figure~\ref{splines}, is denoted as
\begin{equation}\label{basisSL5}
\{S_{-2},S_{-1},S_0,\dots , S_4, S_5\},
\end{equation}
where $S_k(v)=S_0(v-k/5)$ for $k=-2, \dots , 4$ and $S_{5}(v)=c$.

\begin{center}
\includegraphics[scale=0.44]{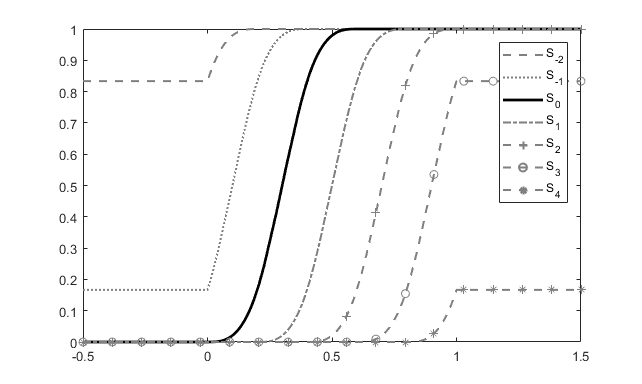}\captionof{figure}{Cubic increasing monotonic base functions.}\label{splines}
\end{center}
Any linear combination with positive scalars of the splines in (\ref{basisSL5}) define a spline which is an increasing and positive function. \\

\textbf{Distortion density as a spline.} Let $\widehat{h}^2_l(v)$ denote an increasing monotone cubic density defined as a linear combination of  $l=L+3$ splines in (\ref{Scbase2L})
\begin{equation} 
\widehat{h}^2_l(v) = \sum_{k=-2}^{L} \lambda_k\cdot  S_k(v),
\end{equation}
where $\lambda_k\geq 0$ for all $k=-2,\dots , L$. Notice that by setting the scalars to be non-negative, $\widehat{h}^2_l $ is increasing. However,   $\widehat{h}^2_l$ must integrate to $1$ on $[0, 1]$, hence
\[\int_0^1 \widehat{h}^2_l(v) \: dv=  \sum_{k=-2}^{L} \lambda_k\cdot \int_0^1 S_k(v)\:dv  =\sum_{k=-2}^{L} \lambda_k\cdot  \left( \sum_{i=1}^nA_{ik}\right) = \sum_{k=-2}^{L} \lambda_k\cdot a_k=1,\]
where
\begin{equation}\label{Aikandsum}
A_{ik}= \int _{\frac{i-1}{n}}^{\frac{i}{n}}S_k(v)\:dv, \quad  a_k=\sum_{i=1}^n A_{ik}.
\end{equation}

\textbf{Prices with the spline.} For each sample $j=1,\dots , m$,  the prices with $\widehat{h}^2_l$ are
\begin{equation}\label{pihatsplinej}
\widehat{\pi}^{(j)} = \sum_{i=1}^n x_{[i]}^{(j)} \cdot \int_{\frac{i-1}{n}}^{\frac{i}{n}} \widehat{h}^2_l (v) \: dv  = \sum_{i=1}^n x_{[i]}^{(j)} \cdot \left( \sum_{k=-2}^{L} \lambda_k \, A_{ik} \right).
\end{equation}

\textbf{Estimation.} Given prices   $\pi^{(j)}$ calculated as in (\ref{directpiaveru}), (\ref{directpisles1}) or (\ref{directpisbigg1}) and the prices calculated in (\ref{pihatsplinej}) for every sample $j=1,\dots , m$, we solve
\[
\begin{aligned}
& \underset{\lambda}{\min}
& &  \sum_{j=1}^m (\hat{\pi}^{(j)}-\pi^{(j)})^2 \\
& \text{s.t.} & & \sum_{k=-2}^{L} \lambda_k\cdot a_k=1 \\
& & &   \lambda_k \geq 0,\: k=-2,\dots , {L},
\end{aligned} \tag{$P_2$}\label{invphd2}
\]
where $a_k$ is defined in (\ref{Aikandsum}). \\
The estimations obtained by solving  (\ref{invphd}) and (\ref{invphd2}) are presented below.
%Once the  inverse problem is defined for both, step and spline cases
%On the following, we solve the problems (\ref{invphd}) and (\ref{invphd2}) in order to recover the distortion density of the $\AVaR$-distortion and power premium principles.\\

{\bf{$\AVaR$ distortion premium.}} We consider particular cases of $h_\alpha$ for $\alpha=0.9, 0.95$. We estimate the distortion density for each of the cases,  with  two different step functions, corresponding to $l=8, 10$ steps, and two  different spline basis functions of dimensions $l=8, 13$, respectively. \\

\textbf{Step function.} The estimated step distortions $\widehat{h}_l$ for $l=8, 10$ are obtained by solving (\ref{invphd})  and illustrated below.

%We calculate the estimated prices as in (\ref{pihatstepj}) and solve (\ref{invphd}) for $l=8, 10$. The  estimated step function $\widehat{h}_l$ is illustrated in the following figures.

\begin{center}
\includegraphics[scale=0.4]{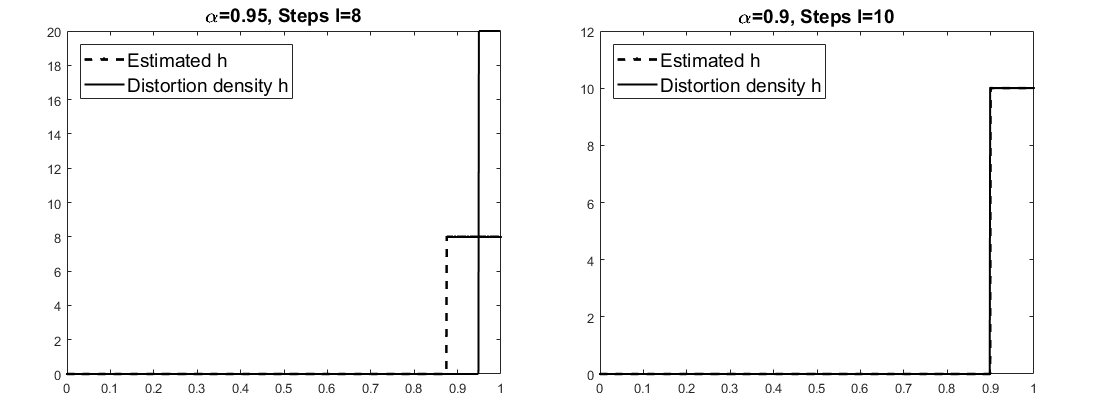}
\captionof{figure}{The true distortion density  $h_{{\alpha}}$ for $\alpha=0.9, 0.95$ and their respective step functions estimators for  $l=8$ steps, and $l=10$ steps.}
\end{center}

\textbf{Splines.} The estimated spline distortions $\widehat{h}^2_l$ for $l=8,13$ are obtained by solving (\ref{invphd2}) and illustrated below.

%We calculate the prices for  $\widehat{h}^2_l$ as in (\ref{pihatsplinej}) and solve (\ref{invphd2}) for $L=5, 10$ subintervals, hence for dimensions $l=8,13$. The  estimated splines  are illustrated in the following figures.

\begin{center}
\includegraphics[scale=0.4]{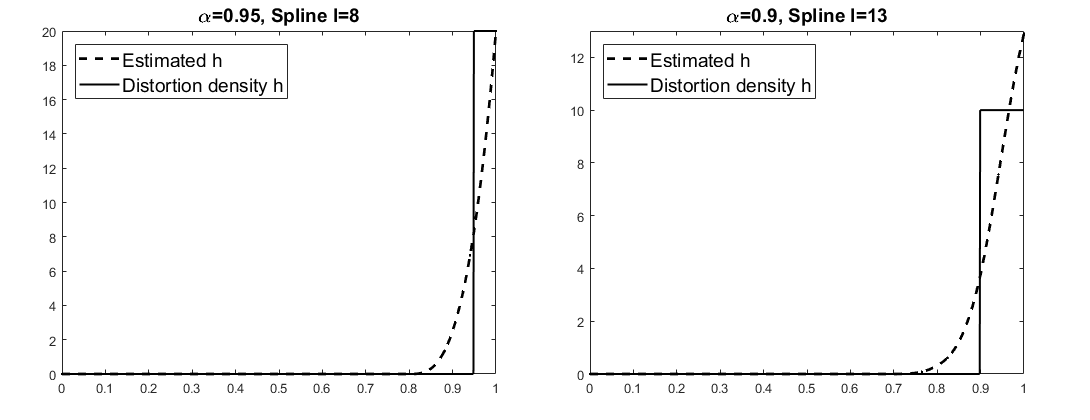}\captionof{figure}{The true distortion density  $h_{{\alpha}}$ for $\alpha=0.9, 0.95$ and their respective spline estimators for $l=8$ and $l=13$ spline base dimension.}
\end{center}

{\bf{Power distortion premium.}} For this case we consider$h^{(s)}$ for $s=0.8, 3$. We  solve (\ref{invphd})  and (\ref{invphd2}) with the same number of steps and number of spline basis functions as before.\\

\textbf{Step function.}  The estimated step distortions  $\widehat{h}^1_l$  for  $l=8,10$ are obtained by solving (\ref{invphd})
and illustrated below.

\begin{center}
\includegraphics[scale=0.4]{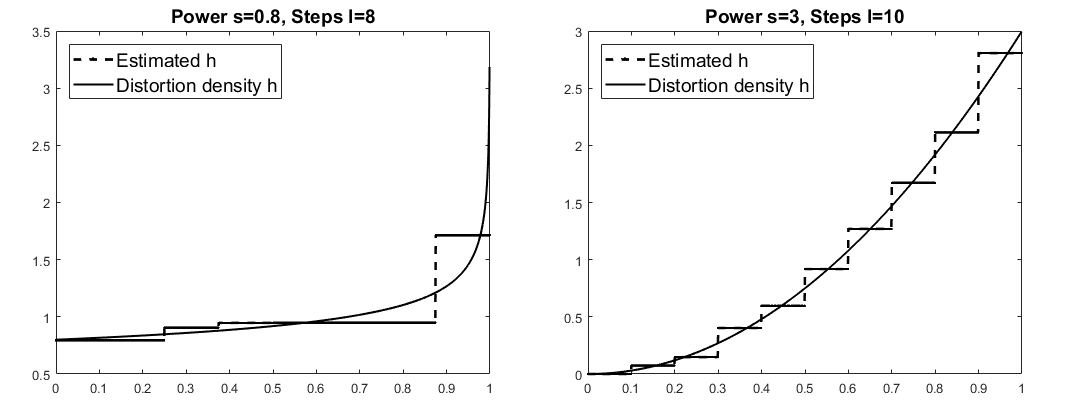}\captionof{figure}{The true distortion density  $h^{(s)}$ for $s=0.8,3$ and their respective estimated step distortions with $l=8,10$ steps.}
\end{center}

\textbf{Splines.} The estimated spline distortions $\widehat{h}^1_l$ for  $l=8,13$  are obtained by solving (\ref{invphd2}) and illustrated below.

\begin{center}
\includegraphics[scale=0.4]{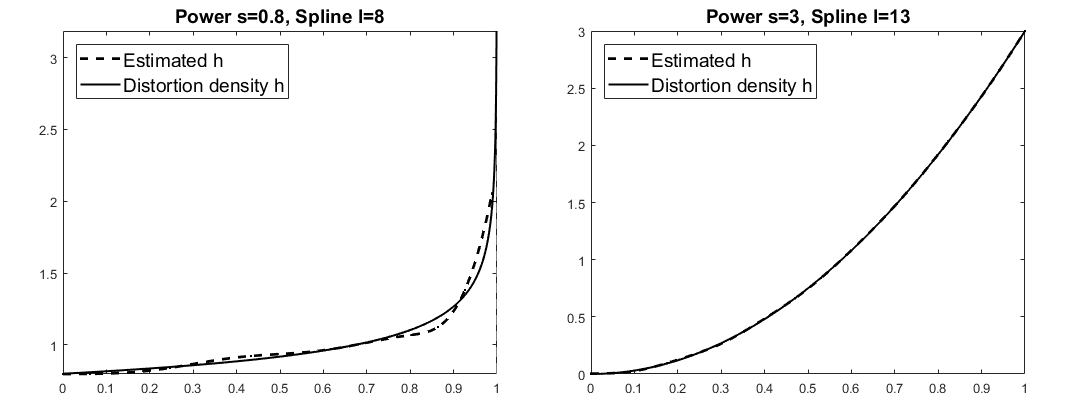}\captionof{figure}{The true distortion density  $h^{(s)}$ for $s=0.8,3$ and their respective estimated spline distortions with $l=8,13$ spline base dimension.}
\end{center}

The optimal values of the optimization problems for all the cases can be seen in the following table.
\begin{table}[h!]
\centering
\label{my-label}
\begin{tabular}{l|l|l}
   $\AVaR$               & $\alpha =0.9$ & $\alpha =0.95$ \\ \hline
\multirow{3}{*}{} Step $l=8$                            &  7.3248 &  107.1562 \\ %\hline%\cline{2-4}
\multirow{3}{*}{} Step $l=10$                   & 0         &  58.4835 \\ %\hline%\cline{2-4}
\multirow{3}{*}{} Spline $l=8$        &  0.0322  & 13.0785 \\% \hline%\cline{2-4}
\multirow{3}{*}{} Spline $l=13$        &  0.0154  &  0.0502 \\ %\cline{2-4}
\end{tabular}
\begin{tabular}{l|l|l}
   Power               & $s =0.8$ & $s =3$ \\ \hline
\multirow{3}{*}{} Step $l=8$                            &   0.0012   &  1.1466e-04 \\ %\hline%\cline{2-4}
\multirow{3}{*}{} Step $l=10$                   & 0            & 5.1772e-05 \\ %\hline%\cline{2-4}
\multirow{3}{*}{} Spline $l=8$        &  3.6976e-04  & 0 \\% \hline%\cline{2-4}
\multirow{3}{*}{} Spline $l=13$        &  1.3251e-04  & 0 \\ %\cline{2-4}
\end{tabular}\caption{Optimal values of the problems (\ref{invphd}) and (\ref{invphd2}) for the $\AVaR-$distortion and the power distortion.}
\end{table}

\section{Ambiguity}\label{sec:6}

In this section we combine the distortion premium with the ambiguity principle. Such an approach allows us to incorporate model uncertainty into the premium. Recall that, by setting the distortion density to $h=1$, we would price just with the ambiguity principle. As was mentioned in Section~\ref{sec:1}, distances can be used to define ambiguity sets. Here, closed Wasserstein balls will serve as ambiguity sets. These sets will be centred at $F$, an initial distribution, that we refer to as our baseline model.

%On the other hand, even when setting $h=1$, i.e. no distortion, but using the ambiguity principle leads to a reasonable pricing rule, the ambiguity price.

\begin{definition}[Robust distortion premium under Wasserstein balls with $d_1$]\label{wassballs} Let $F$ be the baseline loss distribution, $h$ a distortion density. The robust distorted price of order $r\geq 1$ is
\begin{equation}\tag{P-r}\label{P-r}
\pi^\epsilon_{h,r,d_1} (F)=\sup \{ \pi_h(G) : G\in\mathcal{B}_{r,d_1}(F,\epsilon) \} ,
\end{equation}
where $\mathcal{B}_{r,d_1}(F,\epsilon) =\{G:\, WD_{r,d_1}(G, F)\leq \epsilon\}.$ We call the worst case distribution and denote it by $F^*$ if $F^* \in \mathcal{B}_{r,d_1}(F,\epsilon) $ and is such that 
$$\pi^\epsilon_{h,r,d_1} (F) =  \pi_h(F^*).$$
\end{definition}

\begin{remark}\label{orderwd}
Notice that for $r_1 \le r_2$
\begin{equation} 
WD_{r_1,d_1} \le WD_{r_2,d_1},
\end{equation}
thus $\mathcal{B}_{r_1,d_1} \supseteq  \mathcal{B}_{r_2,d_1}$.
\end{remark}

We can say more about the value and solution of (\ref{P-r})  if we choose $r=p$. We start with bounded distortion densities, i.e. for $p=1$ and $q=\infty$.

\begin{proposition}[Characterization of the worst case distribution  for $\mathbf{r\geq p=1}$]\label{boundwd1} Let the baseline distribution $F$ have its first moment finite.
\begin{itemize}
\item[(i)] If $h$ is unbounded, then  (\ref{P-r}) for $r = 1$ is unbounded.
\item[(ii)] If $h$ is bounded with $\sup_v h(v) = \|h\|_\infty$, then   (\ref{P-r}) is bounded for all $r\geq 1$. If $r=1$, the optimal value of (\ref{P-r})  is
$$\pi^\epsilon_{h,1,d_1}(F) = \pi_h(F) + \epsilon \cdot \|h\|_\infty.$$
We interpret the additional term $\epsilon \cdot \|h\|_\infty$ as the {\em ambiguity premium}.
For the worst case distribution,
\begin{itemize}
\item if $h(v) = \|h\|_\infty $ for  $v \ge 1-\eta$ and $0<\eta\leq 1$, then the supremum is attained at
$$F_\eta^*(x) = \left\{ \begin{array}{ll}
F(x)                                      & \qquad \qquad x < F^{-1}(1-\eta),\\
1-\eta                                  & \qquad \qquad F^{-1}(1-\eta)\le x < F^{-1}(1-\eta) +  \epsilon/ \eta ,\\
F\left(x- \epsilon/ \eta  \right)   & \qquad \qquad  x \geq  F^{-1}(1-\eta) + \epsilon/ \eta .
\end{array}
\right.
$$
\item Otherwise, the supremum is not attained, but can be approximated by the sequence $F^*_{1/n}(x)$,  $\forall n\in\mathbb{N}$.
\end{itemize}

\end{itemize}
\end{proposition}
\begin{proof} (i) Given that $h$ is increasing and unbounded,  the increasing sequence $K_n = h\left(1- 1/n \right)$, is such  that $\lim_{n\rightarrow \infty } K_n =\infty$. For all $n\in \mathbb{N}$ we define a distribution $G_n$ such that 
$$G_n^{-1}(v) = F^{-1}(v) +   \epsilon\cdot n\,  \mathds{1}_{[1-1/n , 1]}.$$
$G_n$ is on the boundary of $ \mathcal{B}_{1,d_1}(F,\epsilon)  $  and 
\[ \pi_h(G_n)  =  \pi_h(F) + \epsilon\cdot n \int_{1-1/n}^1 h(v)\, dv  \geq  \pi_h(F) + \epsilon \, K_n. \]
 
Hence, (\ref{P-r}) is unbounded for $r=1$.
(ii) It is sufficient to prove (\ref{P-r}) is bounded for $r=1$ since  $\mathcal{B}_{1,d_1} \supseteq  \mathcal{B}_{r,d_1}$ for all $r\geq 1$ (see Remark~\ref{orderwd}). Any admissible $G$ for $r=1$  can be written as  $G^{-1}(v) = F^{-1}(v) + G_1^{-1}(v)$, where $G_1$ is such that
$\int_0^1 G_1^{-1}(v) \, dv \le \epsilon$. Since $F$ has its first moment finite, the following upper bound is finite:
\begin{equation}\label{optimum}
\pi_h(G) = \pi_h(F) + \int_0^1  G_1^{-1}(v)\, h(v) \, dv \le \pi_h(F) + \epsilon \cdot \|h\|_\infty.
\end{equation}
The distribution $F_\eta^*(x)$ given in  the Proposition has inverse
$$(F_\eta^*)^{-1} (v) = F^{-1}(v) + \frac{\epsilon}{\eta} \mathds{1}_{[1-\eta,1]}.$$
Therefore, $F_\eta^*  $ is on the boundary of $\mathcal{B}_{1,d_1}(F,\epsilon)$ and
\[
\pi_h(F_\eta ^*)=\int_0^1 \left(F^{-1}(v) + \frac{\epsilon}{\eta} \mathds{1}_{[1-\eta,1]}\right) h(v) \, dv= \pi_h(F) + \frac{\epsilon}{\eta} \int_{1-\eta}^1 h(v) \, dv.
\]
If $h(v) = \|h\|_\infty$ for $v \ge 1-\eta$, then $F_\eta^*$ attains the upper bound in (\ref{optimum}).
Otherwise, $F^*_{1/n}$ approaches the maximum from below, since
$$ (F_{1/n}^*)^{-1} (v) = F^{-1}(v) + \epsilon \cdot n \mathds{1}_{[1-1/n,1]}, $$
and
\[
\pi_h(F^*_{1/n})= \pi_h(F) + \epsilon \cdot n  \int_{1-1/n}^1 h(v) \, dv \uparrow \pi_h(F) + \epsilon \cdot \|h\|_\infty.
\]

\end{proof}

\begin{remark} The solution $F^*_\eta$ in Proposition~\ref{boundwd1} is not unique. Any distribution $\tilde{F}_\eta $ such that ${\tilde{F}_\eta}^{-1}(v) = F^{-1}(v) + \frac{\epsilon}{\eta}\cdot k(v)\mathds{1}_{[1-\eta ,1]} $, with $\frac{1}{\eta}\cdot k(v)\mathds{1}_{[1-\eta ,1]}$ a density on $[0,1]$, attains the supremum.
\end{remark}
%$\int_{1-\eta}^1 k(v)\,  dv = \eta$
As an example, we illustrate the worst case distribution for the $\AVaR$ premium.
\begin{center}
\includegraphics[scale=0.4]{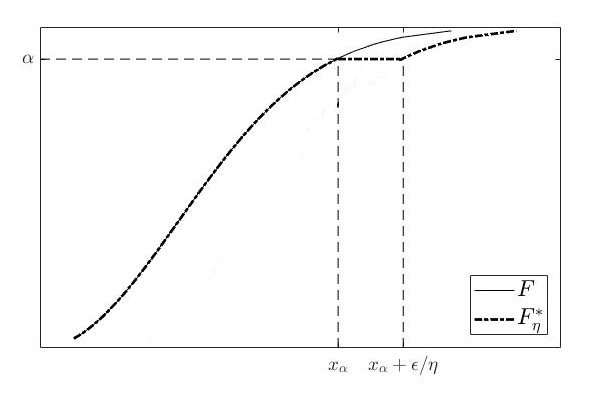}\captionof{figure}{The worst case distribution $F^*_\eta $ for $h_\alpha$ with $\alpha =0.9$   is obtained by shifting   $F$ from $x_\alpha$,   a length $\epsilon /\eta$, where $x_\alpha =F^{-1}(\alpha)$ and $\eta =1-\alpha$.}\label{worstcaseha9}
\end{center}

If  $h$ is unbounded we can characterize the solution of  (\ref{P-r}) as follows.
% may also change the ambiguity ball to obtain finiteness.

%\begin{definition}[Robust distortion premia under Wasserstein balls with $d_p$]\label{wassballs} The robust distorted price for distance $d_p$ is
%\begin{equation}\tag{P-p}\label{P-p}
%\pi^\epsilon_{h,r,d_p} (F)=\sup \{ \pi_h(G) : G\in\mathcal{B}_{r,d_p}(F,\epsilon) \} ,
%\end{equation}
%where $\mathcal{B}_{r,d_p}(F,\epsilon) =\{G:\, WD_{r,d_p}(G, F)\leq \epsilon\}.$
%\end{definition}

\begin{proposition}[Characterization of the worst case distribution  for $\mathbf{r\geq p>1}$]\label{unboundwd1}  Let the baseline distribution $F$ have finite $p-$moments. If  $h\in \mathcal{L}^q$, then  (\ref{P-r})  is bounded for $r\geq p$. If $r=p$, the optimal value of (\ref{P-r})  is
$$\pi^\epsilon_{h, p, d_1}(F)= \pi_h(F) + \epsilon \cdot \|h\|_q^q.$$
Also in this case, the term $\epsilon \cdot \|h\|_q^q$ is interpreted as {\em ambiguity premium}.

Furthermore, the worst case distribution $F^*$ of (\ref{P-r}) for $r=p$ is such that
$${F^*}^{-1}(v) = F^{-1}(v) + \epsilon \cdot \left(\frac{h(v)}{|| h||_q}\right)^{q/p} .$$

\end{proposition}

%the distortion density $h \in \mathcal{L}^q$ and  let $F$ have finite $p$-moments with $p$ and $q$ conjugates. If  the order of the Wasserstein ball is $r\ge p$, then the problem  (\ref{P-r}) $r=p$ is bounded. If $r=p$ then
%$$\sup\{ \pi_h(G): WD_{p,d_1} (F,G)\le \epsilon \} = \pi_h(F) + \epsilon \|h\|_q^q$$

\begin{proof}
We prove (\ref{P-r}) is bounded for $r=p$ and by Remark~\ref{orderwd} we have boundness for all $r\geq p$.  Notice that, for all admissible $G$, if $r=p$, we have
 \begin{align*}
 \int_0^1 G^{-1}(v) \, h(v)\, dv &\leq \int_0^1 F^{-1}(v)\, h(v) \, dv + \int_0^1 |G^{-1} -F^{-1}| \, h(v)\, dv \\
 &\leq \pi_h(F) + \left(\int_0^1  |G^{-1} -F^{-1}|^p \, dv \right) ^{1/p}   ||h||_q   \\
 &\leq \pi_h(F) + \epsilon \cdot  || h||_q.
 \end{align*}
$F^*$ is admissible since it is on the boundary of $\mathcal{B}_{p, d_1}(F, \epsilon)$ 
$$WD_{p, d_1}(F, F^*) =\left( \int_0^1  \epsilon^p \cdot \left( \frac{h(v)}{ || h ||_q}\right) ^q  \, dv \right)^{1/p}=\epsilon,$$ and $F^*$ attains the upper bound 
\[ \pi_h(F^*)  - \pi_h(F) = \int_0^1  \epsilon \cdot \left( \frac{h(v)}{ || h ||_q}\right) ^{q/p} h(v)\, dv =\epsilon \cdot \int_0^1\frac{ h(v)^q}{||h ||^{q-1}_q} \, dv =\epsilon \cdot || h ||_q. \]
\end{proof}

Under some conditions on $h$ we can also prove unboundness of (\ref{P-r}) for $r>p>1$ in the case where $h$ is not in $L^q$, where $q$ is the conjugate of $p$, the finite moments of $F$.
\begin{proposition}[Unboundness for $\mathbf{r>p>1}$]  Let the baseline distribution $F$ have finite $p-$moments and let $h\notin \mathcal{L}^q$, for $p,\, q$ conjugates and $r,\,s$ conjugates with $r>1$. If there exists $s_1<s$ such that $\int_0^1 h(v)^{s_1} \, dv =\infty$   and $h\in \mathcal{L}^t$, for all $t<s_1$, then (\ref{P-r}) is unbounded for all $r>p$. 
\end{proposition}
\begin{proof}
Define  $\psi_\eta(v) = h(v)^{s_1-1}\mathds{1}_{[1-\eta, 1]}$. Since $\psi_\eta \in \mathcal{L}^r$ for $r>1$ (note that $r(s_1 - 1)<s_1$), there exists an $0<\eta<1$ such that 
$$ \int_0^1 \psi_\eta(v)^r\, dv = \int_{1-\eta}^1 h(v)^{r(s_1 - 1)}\, dv < \epsilon.$$
Thus, the distribution $G_\eta$ such that $G_\eta^{-1}(v) = F^{-1} + \psi_\eta(v)$ is in $\mathcal{B}_{r,d_1}(F, \epsilon)$.  And its premium is unbounded
$$ \pi_h(G_\eta) = \pi_h(F) + \int_0^1 \psi_\eta (v)\, h(v)\, dv =  \pi_h(F) + \int _{1-\eta}^1 h(v)^{s_1}> \infty. $$
\end{proof}

\begin{remark} If instead of the metric $d_1$ we consider $d_p(x,y)= |x^p - y^p|$ as underlying metric for the Wasserstein distance, we could define the ambiguity principle 
\begin{equation}\tag{P-dp}\label{P-dp}
\pi^\epsilon_{h,1,d_p} (F)=\sup \{ \pi_h(G) : G\in\mathcal{B}_{1,d_p}(F,\epsilon) \} ,
\end{equation}
where $\mathcal{B}_{r,d_p}(F,\epsilon) =\{G:\, WD_{r,d_p}(G, F)\leq \epsilon\}.$ It is easy to see that, if $F$ has $p-$moments the constraint of the balls make all of admissible distributions to have also $p-$moments, therefore for  Proposition~\ref{contqbigg1}, if $h\in \mathcal{L}^q$ , then (\ref{P-dp}) is bounded. Furthermore, continuity respect to this Wasserstein distance implies our continuity results in Section~\ref{sec:3}.

\end{remark}

\section{Conclusions}

After some introduction about general premium principles we propose generalizations of the distortion premium. In addition, we have studied in detail three  functional relationships for the distortion premium
\begin{itemize}
\item the premium function $F \mapsto \pi_h(F)$, i.e. the properties of $\pi_h$ as a premium principle,
\item the direct function $h \mapsto \pi_h(F)$, i.e. the dependency on the distortion density,
\item the inverse functions $\pi_h(F) \mapsto h$.
\end{itemize}

The smoothness properties are important for robustness aspects, however  it is well known that a quite smooth direct function makes the inverse problem difficult.
We showed however that the inverse problem is identifiable and we gave a simple quadratic optimization problem to estimate it from empirical data. We successfully illustrated this in a simulation study, the application on real data  is left for further research.
We also identified the ambiguity premium for Wasserstein balls as ambiguity sets offering, in some cases, a specific formulation of the worst case distribution. It turned out that the extra premium for ambiguity depends on the distortion function $h$ and  in a multiplicative way on the ambiguity radius $\epsilon$, but does not on the loss distribution $F$ itself. Thus it is the same for all contracts and can be calculated in a separate manner. Finally, by using different distances as underlying metrics for the Wasserstein ball, and hence, for the ambiguity set, we could find bounds for the robust premium is always bounded.

\section*{Appendix}

{\bf Properties of the generalized distortion premium. } We consider here the generalized distortion premium
\begin{equation}\label{generalizedAvar}
R(X)=\int_0^1 \nu(\AVaR_\alpha(X)) \, k(\alpha) \, d\alpha,
\end{equation}
% $\nu(x)\ge x$
where $X \in \mathcal{L}^1$, $\nu$ a convex, monotone Lipschitz function  and $k$ a non-negative weight function on [0,1], which satisfies $\int_0^1 (1-\alpha)^{-1} \, k(\alpha) \, d\alpha < \infty$. Clearly, $X \mapsto R(X)$ is convex and monotone, but is positively homogeneous and/or translation equivariant only if $\nu$ is a multiple of the identity.
To see this, consider the subdifferential of $R$ at $Y \in \mathcal{L}^1$ is
\begin{equation}\label{ZetYpsilon}
Z_Y = \int_0^1 \nu^\prime (AV@R_\alpha(Y)) (1-\alpha)^{-1} \one_{Y > F_Y^{-1}(\alpha)} \, k(\alpha) \, d\alpha \quad \in \mathcal{L}^\infty,
\end{equation}
where $F_Y$ is the distribution function of $Y$.
Notice that $\mathbb{E}(Y \cdot Z_Y)$ depends only on the distribution function $F_Y$. After some calculation, one finds that
$$\mathbb{E}(Y \cdot Z_Y) = \int_0^1 \nu^\prime(AV@R_\alpha(Y)) \cdot AV@R_\alpha(Y)  \, k(\alpha) \, d\alpha.$$

Finally, based on the subdifferential, one gets a dual representation
\begin{eqnarray*}
R(X) &=& \sup_{Y \in \mathcal{L}^1} \{ \mathbb{E}(X \cdot Z_Y) \\
       &-& \int_0^1 [\nu^\prime(AV@R_\alpha(Y)) AV@R_\alpha(Y)  - \nu(AV@R_\alpha(Y))] \, k(\alpha) \, d\alpha \},
\end{eqnarray*}
where $Z_Y$ is given by (\ref{ZetYpsilon}).

It is well known (see Pflug and R{\"o}misch 2007~\cite{pflug2007}) that $R$ is positively homogeneous only if 
$$\int_0^1 [\nu^\prime(AV@R_\alpha(Y)) AV@R_\alpha(Y)  - \nu(AV@R_\alpha(Y))] \, k(\alpha) \, d\alpha=0,$$ 
when it is finite. This implies that $\nu(x)= \gamma \cdot x$ for some $\gamma > 0$. $R$ is translation equivariant, if in addition the expectation of the dual multiplier $Z_Y$ is one, which in happens only if $  \int_0^1  \gamma\, k(\alpha) \, d\alpha=1$. \\

{\bf On  different underlying metrics for the Wasserstein distance. } 
There is a whole family of distances on $\mathbb{R}$, which are generalizations of $d_1$. Set for $x,y \ge 0$,
$d_p(x,y)= |x^p - y^p|$. The Wasserstein distance of order 1 with distance $d_p$ is
$$WD_{1,d_p}(F,G) = \int_0^1 |(F^{-1}(v))^p - (G^{-1}(v))^p| \, dv.$$

\begin{lemma}\label{wdp1boundedbywd11} Notice that for $ p\geq 1$
$$WD_{p,d_1}(F,G) \le [WD_{1,d_p}(F,G)]^{1/p}.$$
\end{lemma}
\begin{proof}
By the subadditivity of $x \mapsto x^p$   on $\mathbb{R}_{\geq 0}$ one has that $|x-y|^p \le |x^p - y^p|$ and therefore
\[
WD_{p,d_1}(F,G) =  \left[ \int_0^\infty |F^{-1}(v) - G^{-1}(v)|^p \, dv \right]^{1/p} \leq  [WD_{1,d_p}(F,G)]^{1/p}.
\]
\end{proof}

\begin{remark}\label{wassersteindistancep}
This argument also shows that if $F$ has finite $p-$moments and if \\
$WD_{1,d_p}(F,G)< \infty$ (and a fortiori if $WD_{p,d_1}(F,G)< \infty$), then
also $G$ has finite $p-$moments. On the other hand, if both $F$ and $G$ have finite $p-$moments, then
$$WD_{1,d_p}(F,G) \le p\cdot WD_{p,d_1}(F,G) (1+\|F^{-1}\|_p^{p-1} + \|G^{-1}\|_p^{p-1})$$
(see Lemma 2.19 in\cite{pflug2014}). Therefore, imposing conditions on $WD_{1,d_p}$ or on  $WD_{p,d_1}$ leads to quite similar results.
\end{remark}

\begin{thebibliography}{}
%
% and use \bibitem to create references. Consult the Instructions
% for authors for reference list style.
%


\bibitem{artzner1999}
Artzner, P. and Delbaen, F. and Eber, J.M. and Heath, D.,
Coherent measures of risk,
Mathematical finance,
9,
3,
203--228,
Wiley Online Library,
(1999)

\bibitem{borch1961}
Borch, K.,
The utility concept applied to the theory of insurance,
ASTIN Bulletin: The Journal of the IAA,
1,
5,
245--255,
Cambridge University Press
(1961)


\bibitem{denneberg1990}
Denneberg, D.,
Premium Calculation: Why Standard Deviation Should be Replaced by Absolute Deviation,
ASTIN Bulletin,
20,
2,
181–190,
(1990)


\bibitem{embrechts1997}
Embrechts, P and Kl{\"u}ppelberg, C and Mikosch, T,
Modelling Extremal Events for Insurance and Finance,
Springer,
Berlin, Heidelberg,
(1997),



\bibitem{furmanetal2017}
Furman, E. and Wang, R. and Zitikis, R.,
Gini-type measures of risk and variability: Gini shortfall, capital allocations, and heavy-tailed risks,
Journal of Banking \& Finance,
83,
70--84,
Elsevier,
(2017)

\bibitem{furman2008}
Furman, E. and Zitikis, R.,
Weighted premium calculation principles,
Insurance: Mathematics and Economics,
42,
1,
459--465,
Elsevier,
(2008)


\bibitem{gilboa1989}
Gilboa, I. and Schmeidler, D.,
Maxmin expected utility with non-unique prior,
Journal of mathematical economics,
18,
2,
141--153,
Elsevier,
(1989)


\bibitem{goovaerts2004}
Goovaerts, M. J. and Kaas, R. and Laeven, R. JA and Tang, Q.,
A comonotonic image of independence for additive risk measures,
Insurance: Mathematics and Economics,
35,
3,
581--594,
Elsevier,
(2004)

\bibitem{gourieroux2006}
Gourieroux, C. and Liu, W.,
Sensitivity analysis of distortion risk measures
(2006)


\bibitem{greselin2018}
Greselin, F. and Zitikis, R.,
From the classical Gini index of income inequality to a new Zenga-type relative measure of risk: A modeller’s perspective,
Econometrics,
6,
1,
4,
Multidisciplinary Digital Publishing Institute,
2018

\bibitem{huber2011}
Huber, P. J.
Robust statistics,
International Encyclopedia of Statistical Science,
1248--1251,
Springer,
(2011)



\bibitem{jouini2006}
Jouini, Ely{\`e}s and Schachermayer, Walter and Touzi, Nizar,
Law invariant risk measures have the Fatou property,
Advances in mathematical economics,
49--71,
Springer
(2006)

\bibitem{kiesel2016}
Kiesel, R{\"u}diger and R{\"u}hlicke, Robin and Stahl, Gerhard and Zheng, Jinsong,
The Wasserstein metric and robustness in risk management,
Risks,
4,
3,
32,
Multidisciplinary Digital Publishing Institute,
(2016)

\bibitem{kusuoka2001}
Kusuoka, S.,
On law invariant coherent risk measures,
Advances in mathematical economics,
83--95,
Springer
(2001)

\bibitem{luan2001}
Luan, C.
Insurance premium calculations with anticipated utility theory,
ASTIN Bulletin: The Journal of the IAA,
31,
1,
23--35,
Cambridge University Press
(2001)

\bibitem{nguyen2012}
Nguyen, H. T. and Pham, U. H. and Tran, H. D.,
On some claims related to Choquet integral risk measures,
Annals of Operations Research,
195,
1,
5--31,
Springer,
(2012)


\bibitem{pflug2014}
Pflug, G. Ch. and Pichler, A.,
Multistage Stochastic Optimization.,
Springer International Publishing, Springer International Publishing Switzerland,
1st,
(2014)

\bibitem{pflug2007}
Pflug, G. Ch. and R{\"o}misch, W.
Modeling, measuring and managing risk,
World Scientific
(2007)

\bibitem{pflug2006subdifferential}
Pflug, G. Ch.,
Subdifferential representations of risk measures,
Mathematical programming,
108,
2-3,
339--354,
Springer,
(2006)


\bibitem{pichler2010}
Pichler, A.,
Distance of probability measures and respective continuity properties of acceptability functionals,
(2010)

\bibitem{pichler2013}
Pichler, A.,
The natural Banach space for version independent risk measures,
Insurance: Mathematics and Economics,
53,
2,
405--415,
Elsevier
(2013)


\bibitem{tsukahara2013}
Tsukahara, H.,
Estimation of distortion risk measures,
Journal of Financial Econometrics,
12,
1,
213--235,
Oxford University Press,
(2013)


\bibitem{vallender1974}
Vallender, S.S.,
Calculation of the Wasserstein distance between probability distributions on the line,
Theory of Probability \& Its Applications,
18,
4,
784--786,
SIAM
(1974)


\bibitem{vinel2017}
Vinel, A. and Krokhmal, P. A.,
Certainty equivalent measures of risk,
Annals of Operations Research,
249,
1-2,
75--95,
Springer,
(2017)
 
\bibitem{neumann1947}
Von Neumann, J., Morgenstern, O.,
Theory of games and economic behavior,
Princeton University Press, Princeton, NJ, US.,
2nd
(1947)


\bibitem{wang2000}
Wang, S. S.
A class of distortion operators for pricing financial and insurance risks,
Journal of risk and insurance,
15--36,
JSTOR
(2000)


\bibitem{wang1995}
Wang, S.
Insurance pricing and increased limits ratemaking by proportional hazards transforms,
Insurance Mathematics \& Economics,
17,
1,
43--54,
(1995)


\bibitem{wang1996}
Wang, S.
Premium calculation by transforming the layer premium density,
ASTIN Bulletin: The Journal of the IAA,
71--92,
Cambridge University Press
(1996)



\bibitem{wangetal1997}
Wang, S. S. and Young, V. R. and Panjer, H. H.,
Axiomatic characterization of insurance prices,
Insurance: Mathematics and Economics,
21,
2,
173-183,
(1997)
http://EconPapers.repec.org/RePEc:eee:insuma:v:21:y:1997:i:2:p:173-183

\bibitem{wozabal2012}
Wozabal, D.,
A framework for optimization under ambiguity,
Annals of Operations Research,
193,
1,
21--47,
Springer,
(2012)


\bibitem{wozabal2014}
Wozabal, D.,
Robustifying convex risk measures for linear portfolios: A nonparametric approach,
Operations Research,
62,
6,
1302--1315,
INFORMS,
(2014)


\bibitem{yaari1987}
Yaari, M. E.,
The dual theory of choice under risk,
Econometrica: Journal of the Econometric Society,
95--115,
JSTOR
(1987)

\bibitem{young2014}
Young, V. R.,
Premium Principles.,
Wiley StatsRef: Statistics Reference Online,
(2014)

\bibitem{zwillinger2002} 
Zwillinger, D.,
CRC standard mathematical tables and formulae,
Chapman and Hall/CRC,
(2002)

%\bibitem{RefJ}
%% Format for Journal Reference
%Author, Article title, Journal, Volume, page numbers (year)
%% Format for books
%\bibitem{RefB}
%Author, Book title, page numbers. Publisher, place (year)
%% etc
\end{thebibliography}
\end{document}